\documentclass[12pt]{amsart}
\usepackage[foot]{amsaddr}
\usepackage[a4paper,oneside,DIV12]{typearea}
\usepackage{amssymb}
\usepackage{xspace}
\usepackage{ifthen}
\usepackage{url}
\usepackage{enumerate}                     

\usepackage{soul}
\usepackage{tikz}
\usetikzlibrary{arrows,matrix,decorations,arrows,shapes}
\tikzstyle vertex=[circle, line width=0.2mm, draw, fill=black, inner sep=0mm, minimum width=1mm]%
\tikzstyle{edge} = [draw]

\usepackage[ruled,noline,titlenumbered,linesnumbered,noend,norelsize]{algorithm2e}
\DontPrintSemicolon
\SetNlSty{footnotesize}{}{}
\IncMargin{0.5em}
\SetNlSkip{1em}
\SetKwIF{If}{ElseIf}{Else}{if}{then}{else if}{else}{endif}
\SetKw{Return}{return}
\SetKw{Let}{let}

\newenvironment{algorithm-hbox}{\hbadness=10000\begin{algorithm}}{\end{algorithm}}

\renewcommand{\phi}{\varphi}
\renewcommand{\leq}{\leqslant}
\renewcommand{\geq}{\geqslant}

\usepackage{marginnote}

\newtheorem{theorem}{Theorem}[section]
\newtheorem{prop}[theorem]{Proposition}

\theoremstyle{definition}
\newtheorem{problem}{Problem}
\newtheorem{claim}[theorem]{Claim}

\newcommand{\abs}[1]{\left\vert#1\right\vert}

\newcommand{\Alg}[0]{\textsf{BALANCED}\xspace}
\newcommand{\Br}{Builder\xspace}
\newcommand{\Sr}{Scheduler\xspace}

\DeclareMathOperator{\Val}{val}
\DeclareMathOperator{\bal}{bal}

\newcommand{\eqB}{\Psi}

\newcommand{\vr}[1]{\mathbf{#1}}

\allowdisplaybreaks[1]
\begin{document}


\title{A lazy approach to on-line bipartite matching}

\author[J.\ Kozik]{Jakub Kozik} 
\email{Jakub.Kozik@tcs.uj.edu.pl}%
\thanks{Research of J.\ Kozik was supported by Polish National Science Center UMO-2011/03/D/ST6/01370.}

\author[G.\ Matecki]{Grzegorz Matecki}
\email{Grzegorz.Matecki@tcs.uj.edu.pl}%
\thanks{Research of G.\ Matecki was supported by Polish National Science Center 2011/03/B/ST6/01367.}
\address{
  Theoretical Computer Science\\
  Faculty of Mathematics and
  Computer Science\\
  Jagiellonian University in Krak\'{o}w\\
  Poland}%



\begin{abstract}
We present a new approach, called a lazy matching, to the problem of on-line matching on bipartite graphs.
Imagine that one side of a graph is given and the vertices of the other side are arriving on-line.
Originally, incoming vertex is either irrevocably matched to an another element or stays forever unmatched.
A lazy algorithm is allowed to match a new vertex to a group of elements (possibly empty) and afterwords, forced against next vertices, may give up parts of the group. 
The restriction is that all the time each element is in at most one group.
We present an optimal lazy algorithm (deterministic) and prove that its competitive ratio equals $1-\pi/\cosh(\frac{\sqrt{3}}{2}\pi)\approx 0.588$.
The lazy approach allows us to break the barrier of $1/2$, which is the best  competitive ratio that can be guaranteed by any deterministic algorithm in the classical on-line matching.
\end{abstract}

\maketitle

\section{Introduction}
Many problems of task-server assignment can be modeled as finding a matching in a bipartite graph $G=(U,D,E)$. 
Vertices of one part (set $D$) correspond to servers, and the vertices of the other (set $U$) -- to tasks. 
An edge between a task and a server indicates that the server is capable of performing the task. 
In a simple  setting, when one server can realize at most one task, the problem of maximization of the number of realized tasks reduces to finding a maximum matching.
In real-life applications it is very common that not all tasks are known a priori and some decisions about assignments have to be taken with no knowledge about future tasks. 
A simple model for this situation is on-line bipartite matching. 
In this setting servers are known from the beginning and tasks are revealed one by one. 
The decision about assignment of each task has to be made just after its arrival and cannot be changed in the future.   
Suppose that there are $n$ servers, $n$ tasks are going to be revealed, and capabilities of servers are such that it is possible to realize all the tasks.
(i.e.\ there exists a perfect matching in the tasks-servers graph). 
It is an easy exercise to show that even with these restrictions it is possible to present tasks in such a way, that the constructed assignment is at most $\lceil n/2 \rceil$. 
On the other hand, any greedy assignment strategy guarantees that at least half of the tasks will be assigned. 

In their classical contribution Karp, Vazirani, and Vazirani~\cite{KVV90} take other approach in which the graph to be presented is fixed before the first task is presented. 
In particular the presented graph does not depend on the decisions of the assigning algorithm. 
It does not make any difference for the worst-case analysis of the algorithm, but the approach provides a framework for analyzing randomized ones. 
The authors presented a randomized algorithm which on the average constructs a matching of size at least $(1-1/e)n$. 
They also argued that the result is asymptotically best possible for any randomized algorithm (the original paper~\cite{KVV90} contained a mistake, which has been corrected in~\cite{GM2008}, see also a simplified exposition~\cite{BM2008}). 
The approach of~\cite{KVV90} has been applied to many variants of the original problem, with various practical applications (switch routing problem~\cite{AC2006, AR2003}, on-line auctions~\cite{MGS2011}, Adwords problem~\cite{DH2009, GM2008, MSV2007} etc.) 
Recently a lot of interest is put into a problem of on-line stochastic matching~\cite{BK2010, FMMM2009, KMT2011, MGS2011, MP2012} where a competitive ratio can be  greater than $1-\frac{1}{e}$.
A different approach (called $b$-matching) is presented in~\cite{KalKir2000} where authors allow a server to realize up to $b$ tasks at the same time. 
They showed an optimal deterministic algorithm with  competitive ratio $1-\frac{1}{(1+\frac{1}{b})^b} $ (which tends to $1-\frac{1}{e}$ with $b \rightarrow \infty$).

It is not hard to imagine a situation that the cost of a running server is roughly the same as the cost of an idle server. 
It might be profitable to start to realize some task on many servers simultaneously.
In this work, we allow our algorithm to assign more than one server to an incoming task, and later on gradually giving up its computations and switch some servers to yet another tasks.
Clearly, it is sufficient that at least one server completes the computations of a task.
This enables to postpone the decision about which server is going to accomplish a given task. 
With this settings we present an algorithm matching at least $(1-\pi/\cosh(\frac{\sqrt{3}}{2}\pi))n+o(n)\approx 0.588n+o(n)$ vertices where $n$ is the maximum size of a matching in the whole presented bipartite graph.

The lazy approach was first introduced by Felsner in~\cite{Fel97} 
as an adaptive generalization of the on-line chain partitioning problem. 
His adaptive on-line modification is in our terms a lazy approach to chain partitioning of posets (see also~\cite{Kloch07}).
It is still open and seems challenging to verify if adaptive (lazy) approach to chain partitioning allows more efficient on-line algorithms.

\subsection{Related Work}
The problem of $b$-matching seems similar to lazy matching only the roles of servers and tasks are switched. 
However, in the lazy approach an algorithm always ends up with matching of type one to one, while a $b$-matching algorithm can assign one server to many tasks.

Another similar approach was proposed by Feldman et al.\ ~\cite{FMMMa2009} as \emph{free disposal}.
They consider the weighted matching problem where each incoming vertex $u\in U$ may be assigned to one of its neighbors or left alone.
Each vertex $d\in D$ accepts at most $n(d)$ vertices from $U$ with highest-weighted edges.
Here roles of servers and tasks are switched. 
All tasks are given at once and servers are incoming on-line. 
Each server has to be assigned to at most one task. 
In the end, each task chooses at most $n(d)$ servers from all the servers assigned to it  -- the ones with the highest-weighted edge.
The main difference from the lazy approach is that once a connection between a server and a task is established it cannot be changed till the very end.
There is no such restriction in the lazy approach -- a server may drop its task and take a new one during the on-line process.

The idea of dropping an edge from already constructed matching
is investigated in the \emph{preemptive model}. 
Here, edges with weights are incoming on-line and algorithm is allowed to remove previously accepted edges in order to add a new one.
A collection of results on preemptive matching can be found in~\cite{CTV2015,ELSW2013}.

\subsection{Problem definition}\label{sec:definition}

For a positive integer $\alpha$, the \emph{$\alpha$-lazy matching game} is played in rounds between \Sr and \Br. 
They play on a set of vertices $D$, known in advance.
There is finite number of rounds and in each round:
\begin{enumerate}
	\item \Br presents a vertex $u$ and reveals its neighbors $N(u)\subseteq D$.
	\item \Sr assigns to $u$ a set $m(u)\subseteq N(u)$ of size at most $\alpha$ \label{game:2}\\
	  and updates $m(x)=m(x)\setminus m(u)$ for every vertex $x$ presented before.
\end{enumerate}
Let $U$ be the set of all vertices presented by \Br in the game and let $G=(U,D,E)$ be an underlying bipartite graph with $(u,d)\in E$ when $d\in N(u)$, for all $u\in U, d\in D$.
Now, the \emph{size} $n$ of the game is the maximum size of a matching in $G$.
{If either $\alpha\geq \abs{D}$ for all $n$ or more generally $\alpha{\buildrel n\to\infty\over\longrightarrow}\infty$ then the game is called \emph{$\infty$-lazy matching} or simply \emph{lazy matching}.} We refer to this by writing $\alpha\to\infty$.

The goal of \Sr is to maximize the number of nonempty sets $m(u)$ over all vertices $u\in U$. 
Intuitively, such vertex $u$ is successfully matched with arbitrary chosen $d\in m(u)$. 
We  refer to this number as \emph{the size of the matching constructed by Scheduler}.
The goal of \Br is just the opposite, he disturbs \Sr as much as he can.

The interpretation of $\alpha$-lazy game into servers-tasks assignment is clear: $D$ is the set of servers, $U$ is the set of incoming tasks
and algorithm assigns servers in $m(u)$ to an incoming $u$ (possibly canceling previous computations on servers in $m(u)$).
The quality of Scheduler is measured by the number (or the fraction) of tasks which are being realized at the end of the game. 
Note that for $\alpha=1$ the game reduces to the classical on-line bipartite matching.

Let $\mathcal{A}$ be an algorithm which assigns incoming tasks. 
We denote by $\Val_{\mathcal{A}}(n)$ the worst case value of the matching constructed by $\mathcal{A}$ in all possible games of size $n$. 
The value of the $\alpha$-lazy matching problem $\Val_\alpha(n)$ is the maximum value of $\Val_{\mathcal{A}}(n)$ among all $\alpha$-assigning algorithms $\mathcal A$. 
Since no algorithm produces matching larger that $n$ we additionally use a \emph{competitive ratio} defined as
$\liminf_{n\to\infty}{\Val_\alpha(n)}/{n}$.

\subsection{Main results}
To solve the problem of $\alpha$-lazy on-line matching we consider a deterministic algorithm, called $\alpha$-\Alg{}.
The algorithm is described in Section~\ref{sect:best_alg}. 
It behaves greedy, i.e., no task is rejected if there is a possibility to run it, and tries to locally balance the sizes of all $m(u)$.
We prove $\alpha$-\Alg{} algorithm is the best possible one.
\begin{theorem}\label{thm:optimal-bal}
  $\alpha$-\Alg{} is an optimal strategy for \Sr in the $\alpha$-lazy matching game.
\end{theorem}

The next two sections are to prove the theorem.
The proof is split into two parts.
The following schema of system of inequalities is crucial for both arguments:
\begin{align}\label{ineq:main-res}
 \left\{
 \begin{array}{l}
 (1+\alpha)\vr{x_0} \leq n,\\
  (\vr{x_0}+\ldots+\vr{x_i})(1+\vr{x_i})\leq n-i,\quad i=1,\ldots,k,\\
  \vr{x_1} \geq \vr{x_2}\geq \ldots \geq \vr{x_k} \geq 0,\\
  \vr{x_0}+\ldots + \vr{x_k} \geq 0.
 \end{array}
 \right. 
\end{align}

We are going to work with $n$ and $\alpha$ fixed. 
Then, the schema is parametrized by a positive integer $k$.
We say that a pair $(k,x)$ satisfy system~\eqref{ineq:main-res} if $x=(x_0,x_1, \ldots, x_k)$ is an integer vector satisfying instance of the schema for this particular $k$.

In Section~\ref{sect:worst_adapive}, we prove (Proposition~\ref{prop:spoiler_bound}) that for every solution $(k,x)$ of~\eqref{ineq:main-res}, every $\alpha$-lazy algorithm $\mathcal A$ can be cheated by some \Br's strategy in a game of size $n$ so that $\mathcal A$ matches at most $n-(x_0+\ldots+x_k)$ vertices.
On the other hand, in Section~\ref{sect:best_alg}, we show (Proposition~\ref{prop:balanced_bound}) that when $\alpha$-\Alg{} constructs a matching of size $k$ in a game of size $n$, then $k=n-(x_0+\ldots+x_k)$ for some $(k,x)$ satisfying~\eqref{ineq:main-res}.
These two facts ensure $\alpha$-\Alg{} is the optimal strategy for \Sr.

In order to determine the competitive ratio of $\alpha$-\Alg{} we need to maximize the sum $x_0+\ldots+x_k$ over all feasible solutions $(k,x)$ of~\eqref{ineq:main-res}.
In Section~\ref{sect:comp} we present the linear programming formulation of~\eqref{ineq:main-res} which, with the use of the Complementary Slackness Theorem, can be easily solved. Finally, we prove
\begin{theorem}\label{thm:ratio}
 The competitive ratio of $\alpha$-lazy on-line matching problem on bipartite graphs 
 (and the competitive ratio of 
 $\alpha$-\Alg{} algorithm)  equals
$1 - \frac{\alpha}{1+\alpha}\prod_{i=1}^{\alpha-1}\frac{i+i^2}{1+i+i^2}$. For $\alpha\to\infty$ it converges to $1 - {\pi}/{\cosh \frac{\sqrt 3 \pi}{2}}\approx 0.588$.
\end{theorem}
The ratio converges fast and for $\alpha\in\{2,3\}$ it is almost optimal, i.e., $5/9\approx 0.556$ and $4/7\approx 0.571$, respectively.

\section{Worst case scenario for a lazy algorithm}
\label{sect:worst_adapive}
Inequalities~\eqref{ineq:main-res} allow $x_0$ to be negative. 
We start with an observation
that in order to maximize the sum $x_0+\ldots+x_k$ it suffices 
to consider solutions of~\eqref{ineq:main-res} with $x_0=\left\lfloor\frac{n}{1+\alpha}\right\rfloor$.

\begin{prop}\label{prop:positiv:x0}
For any pair $(k,x)$ 
satisfying~\eqref{ineq:main-res} there exists a pair $(k',x')$ 
satisfying~\eqref{ineq:main-res} such that $x'_0 = \left\lfloor\frac{n}{1+\alpha}\right\rfloor$ and $x'_0+\ldots+x'_{k'} \geq x_0+\ldots +x_{k}$.
\end{prop}
\begin{proof}
The case when $x_1=0$ is easy: 
just put $x'_0 = \left\lfloor\frac{n}{1+\alpha}\right\rfloor$, $k'=1$ and $x'_1=0$.
We assume that $x_1>0$. 
The claim is proved by induction on $\left\lfloor\frac{n}{1+\alpha}\right\rfloor - x_0$.
The base, when $x_0= \left\lfloor\frac{n}{1+\alpha}\right\rfloor$, is obvious.
For the induction step, let us assume that $x_0<\left\lfloor\frac{n}{1+\alpha}\right\rfloor$.
Let $j$ be the greatest index for which $x_j=x_1$.
Consider the following sequence $(x'_0,\ldots,x'_k)$: $x'_0 = x_0+1$, $x'_j = x_j-1$ and $x'_i = x_i$ for all $i\notin\{0,j\}$.
For $i\geq j$ we have $x'_0+\ldots+x'_i = x_0+\ldots+x_i$.
Therefore 
$$
  (x'_0+\ldots+x'_i)(1+x'_i)\leq (x_0+\ldots+x_i)(1+x_i) \leq n-i.
$$
For  $0<i<j$ we get
\begin{align*}
 (x'_0+\ldots+x'_i)(1+x'_i) &=  (x_0+\ldots+x_{i}+1)(1+x_j) \leq\\
			     &\leq (x_0+\ldots+x_{j-1}+x_j)(1+x_j)\leq \\
			     &\leq n-j \leq n-i.
\end{align*}
Since $x'_0\leq \left\lfloor\frac{n}{1+\alpha}\right\rfloor$ sequence $(k,(x'_0,\ldots,x'_k))$ satisfies~\eqref{ineq:main-res}.
Finally, $\left\lfloor\frac{n}{1+\alpha}\right\rfloor - x'_0 <  \left\lfloor\frac{n}{1+\alpha}\right\rfloor - x_0$, so by the induction hypothesis 
there exists a solution of~\eqref{ineq:main-res} satisfying the claim.
\end{proof}

\begin{prop}\label{prop:spoiler_bound}
For any pair $(k,x)$ 
satisfying~\eqref{ineq:main-res} there exists a strategy for \Br in the $\alpha$-lazy matching game of size $n$ such that any \Sr constructs a matching of size at most $n-(x_0+\ldots+x_k)$.
\end{prop}
\begin{proof}  
By Proposition~\ref{prop:positiv:x0} it is enough  to consider pairs with $x_0\geq0$.
Without loss of generality we assume that $x_k>0$ and describe a strategy for \Br, that does not allow \Sr to construct a matching larger than $n-(x_0+\ldots+x_k)$.
During the game \Br presents a bipartite graph $G=(U, D,E)$ with $\abs{U}=\abs{D}=n$ and maintains an auxiliary structure:
a partition of $U=  U_0\cup U_1\cup\ldots\cup U_k\cup R$ and a partition of $D =  D_0\cup D_1\cup\ldots\cup D_k\cup S$ such that
\begin{align}
    \abs{U_0} &= \abs{D_0} = x_0,\nonumber\\
    \abs{U_i} &=  \abs{D_i}=1+x_i,\qquad\textrm{for }i=1,\ldots,k,\nonumber\\
    N(u_i)    &=  D-(D_0\cup\ldots\cup D_{i-1}),\qquad\textrm{for each }u_i\in U_i,\label{atack:neighbourU}\\
    N(r)      &=  S,\qquad\textrm{for each }r\in R.\label{atack:neighbourR}
\end{align}
Observe that~\eqref{ineq:main-res} guarantees that $x_0+\ldots+x_k\leq n-k$ and thus 
$\sum_{i=0}^k\abs{U_i} = \sum_{i=0}^k\abs{D_i}\leq x_0 + \sum_{i=1}^k(1+x_i) \leq n$.
Therefore $\abs{R}=\abs{S}\geq0$.
It is straightforward that any bipartite graph which can be partitioned in such way contains a perfect matching.

The strategy of \Br is divided into $k+2$ phases enumerated from $0$ to $k+1$.
{Figure~\ref{fig:attack} depicts the evaluation of the strategy described below for $\alpha=2$, $n=18$ and $k=2$.}
In the beginning of the $i$-th phase ($0\leq i\leq k$) sets $U_j$ and $D_j$, for $j<i$, are already fixed. 
Next, during the $i$-th phase, \Br presents $1+x_i$ vertices, or $x_0$ vertices when $i=0$, which form  set $U_i$ with neighborhoods defined by~\eqref{atack:neighbourU}. 
Most important, \Br chooses in the special way the set $D_i\subseteq D-(D_0\cup\ldots D_{i-1})$ of size $1+x_i$ when $i>0$ and of size $x_0$ when $i=0$.
This will conclude phase $i$.
At the very end, within phase $k+1$, \Br presets a set $R$ of size $n-k-(x_0+\ldots+x_k)$ with vertices neighbouring with all vertices in $S=D-\bigcup_{i=0}^kD_i$.
  
It remains to define \Br's choice of $D_i$.
For that, after each phase $i$ ($0\leq i\leq k$) \Br maintains the following:
\begin{itemize}
    \item[($\star$)] 
    there are $i$ distinct vertices $y_1,\ldots,y_i\in U_0\cup\ldots\cup U_i$ such that
    $m(y)\cap \bigcup_{j=0}^i D_j = \emptyset$ for all $y\in (U_0\cup\ldots\cup U_i)\setminus \{ y_1,\ldots,y_i\}$.
\end{itemize}
Observe that, for a fixed $X\subseteq D$ and a fixed $y\in U$,
once the condition $X\cap m(y) = \emptyset$ is satisfied, it will stay so to the end of the game
as $m(y)$ may only shrink later on in the game.
We prove by induction that choosing such $D_i$ is possible in every phase by successively finding correct $y_i$'s. 

For $i=0$, \Br chooses\footnote{Let $m(X) = \bigcup_{x\in X}m(x)$ for a set $X$.} any $D_0\subseteq D-m(U_0)$ of size $x_0$ and it is possible as 
$\abs{D-m(U_0)}\geq n-\alpha x_0\geq x_0$ (because $\abs{m(y)}\leq\alpha$ for all $y\in U_0$, and by~\eqref{ineq:main-res}).
This way $D_0\cap m(y)=\emptyset$ for all $y\in U_0$ as required.

\begin{figure}[tbh]
\begin{center}
\begin{tikzpicture}
 \foreach \p in {1,...,18} {
  \coordinate (d\p) at (\p*0.7,-1);
  \fill (d\p) circle (1pt) node[below] {\footnotesize $~\atop\p$};
  }
  
 \foreach \q in {1,...,18} {
  \coordinate (u\q) at (\q*0.7,1);
 }
 \fill (u1) circle(1pt) node[above] {\footnotesize $7\atop 8$};
 \fill (u2) circle(1pt) node[above] {\footnotesize $9\atop10$};
 \fill (u3) circle(1pt) node[above] {\footnotesize $\textrm{\st{$11$}}\atop12$};
 \fill (u4) circle(1pt) node[above] {\footnotesize $\textrm{\st{$13$}}\atop14$};
 \fill (u5) circle(1pt) node[above] {\footnotesize $\textrm{\st{$15$}}\atop16$};
 \fill (u6) circle(1pt) node[above] {\footnotesize $\textrm{\st{$17$}}\atop18$};
 \draw[rounded corners] (u1) ++ (-0.3,-0.2) rectangle  (6*0.7+0.3,1.8) ;
 \draw[rounded corners] (d1) ++ (-0.3,-0.6) rectangle  (6*0.7+0.3,0.2-1) ;
 
 \fill (u7) circle(1pt) node[above] {\footnotesize $~\atop17$};
 \fill (u8) circle(1pt) node[above] {\footnotesize $~\atop15$};
 \draw[rounded corners] (u7) ++ (-0.3,-0.2) rectangle  (8*0.7+0.3,1.5) ;
 \draw[rounded corners] (d7) ++ (-0.3,-0.6) rectangle  (8*0.7+0.3,0.2-1) ;
 
 \fill (u9) circle(1pt) node[above] {\footnotesize $~\atop13$};
 \fill (u10) circle(1pt) node[above] {\footnotesize $~\atop11$};
 \draw[rounded corners] (u9) ++ (-0.3,-0.2) rectangle  (10*0.7+0.3,1.5) ;
 \draw[rounded corners] (d9) ++ (-0.3,-0.6) rectangle  (10*0.7+0.3,0.2-1) ;
 
 \foreach \q in {11,...,18} {
   \fill (u\q) circle(1pt);
 }
 \draw[rounded corners] (u11) ++ (-0.3,-0.2) rectangle  (18*0.7+0.3,1.2) ;
 \draw[rounded corners] (d11) ++ (-0.3,-0.6) rectangle  (18*0.7+0.3,0.2-1) ;
 
 \path (u1) ++ (2,1.1) node {$U_0$} ;
 \path (d1) ++ (2,-1) node {$D_0$} ;
 
 \path (u7) ++ (0.4,0.8) node {$U_1$} ;
 \path (d7) ++ (0.4,-1) node {$D_1$} ;
 
 \path (u9) ++ (0.4,0.8) node {$U_2$} ;
 \path (d9) ++ (0.4,-1) node {$D_2$} ;
 
 \path (u11) ++ (2.4,0.5) node {$R$} ;
 \path (d11) ++ (2.4,-1) node {$S$} ;
 
 \draw (u1) ++ (-0.2,-0.2) rectangle (1*0.7+0.2,2) ;
 \path (0.8,2.2) node {$y_1$};
 
 \draw (u2) ++ (-0.2,-0.2) rectangle (2*0.7+0.2,2) ;
 \path (2*0.7+0.1,2.2) node {$y_2$};

 \foreach \q in {u1,u7,u9,u11}  {\path (\q) ++ (-.2,-.21) coordinate (p\q);}
 \foreach \q in {u6,u8,u10,u18} {\path (\q) ++ (.2,-.21) coordinate (p\q);}
 \foreach \q in {d1,d7,d9,d11}  {\path (\q) ++ (-.2,.21) coordinate (p\q);}
 \foreach \q in {d6,d8,d10,d18} {\path (\q) ++ (.2,.21) coordinate (p\q);}

 \filldraw[black!50!white] (pu11) -- (pd11) -- (pd18) -- (pu18) ;
 \filldraw[black!40!white] (pu9) -- (pu10) .. controls (10*.7+.2,.1) .. (pd18) -- (pd9);
 \filldraw[black!30!white] (pu7) -- (pu8) .. controls (8*.7+.2,.1) .. (pd18) -- (pd7);
 \filldraw[black!20!white] (pu1) -- (pu6) .. controls (6*.7+.2,.1) .. (pd18) -- (pd1);

\end{tikzpicture}
\caption{An example of the construction for solution $x=(6,1,1)$ of~\eqref{ineq:main-res} with $n=18$ and $\alpha=2$. 
Recall that $N(U_0)=D$, $N(U_1) = D-D_0$, $N(U_2)=D_2\cup S$ and $N(R)=S$. 
The numbers under vertices in the bottom part represent vertices' names, 
while the numbers in the upper part correspond to sets $m(u)$ (e.g.\ $m(y_1) = \{7,8\}$).
}\label{fig:attack}
\end{center}
\end{figure}
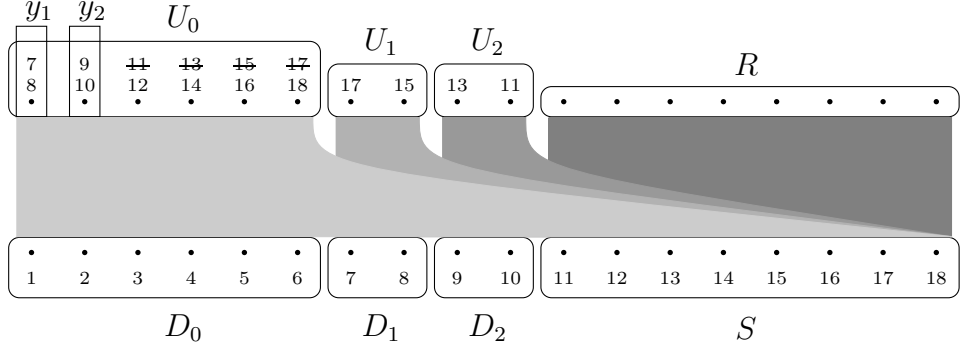

For $1\leq i\leq k$, we assume that ($\star$) holds with vertices $Y=\{y_1,\ldots,y_{i-1}\}$ after $(i-1)$-th phase.
Let $U'=U_0\cup\ldots\cup U_{i}$ and $D'=D-(D_0\cup\ldots\cup D_{i-1})$.
First, note that $m(U'-Y)\subseteq D'$ by the ($\star$)-property after $(i-1)$-th phase.
We split vertices in $D'$ into two blocks:
$D' = m(U'-Y)\cup X$, where $X$ is simply the set of all unmatched vertices in $D'$.
Next, we give a lower bound on the average size of $m(u)$ for $u\in U'-Y$:
\begin{align*}
\frac{\abs{m(U'-Y)}}{\abs{U'-Y}} &=
    \frac{\abs{D'}-\abs{X}}{\abs{U'}-\abs{Y}}  
    =\frac{n-i-\abs{X}-(x_0+\ldots+x_{i-1}-1)}{x_0+\ldots+x_i+1}\\
    &\buildrel \eqref{ineq:main-res}\over \geq 
    \frac{(x_0+\ldots+x_i)(1+x_i)-\abs{X}-(x_0+\ldots+x_{i-1}-1)}{x_0+\ldots+x_i+1}\\
    &= x_i + \frac{1-\abs{X}}{x_0+\ldots+x_i+1} > x_i - \abs{X}.
  \end{align*}
Clearly, there must be a vertex $y_i\in U'-Y$ with $\abs{m(y_i)}\geq \frac{\abs{m(U'-Y)}}{\abs{U'-Y}} > x_i-\abs{X}$.
Since all values are integers we have $\abs{X}+\abs{m(y_i)}\geq 1 + x_i$.
\Br picks $D_i$  to be any subset of $X\cup m(y_i)$ of size $1+x_i$,  to keep the property ($\star$) satisfied after the $i$-th phase.

By the condition ($\star$) after the $k$-th phase there is $Y=\{y_1,\ldots, y_k\}$ such that  $D-S$ and $m(y)$ are disjoint for all $y\in U-Y$.
Therefore, whenever $m(u)\neq\emptyset$,  for $u\in U$, then either $u\in Y$ or $m(u)\subseteq S$. 
It means that the number of such $u$'s is at most $\abs{Y}+\abs{S} = k + n - (\abs{D_0}+\ldots+\abs{D_k}) = n-(x_0+\ldots+x_k)$. 
Consequently, the size of the matching produced by  \Sr is at most $n-(x_0+\ldots+x_k)$.
\end{proof}

\section{The best matching algorithm}
\label{sect:best_alg}

At any moment during a lazy matching game, we say that $d\in D$ is \emph{available} for $u\in U$ if $d\in N(u)$ and $m(x)\neq\{d\}$ for any $x$ presented earlier.
Also, $d$ is \emph{strongly available} for $u$ if it is available for $u$ and $d$ does not belong to any $m(x)$. 
Vertex $e\in U$ is \emph{ready} for $u$ if $m(e)$ contains an  element which is available for $u$.

We present an algorithm for \Sr called $\alpha$-\Alg{}. 
Suppose that vertex $u$ has just been presented and let $U$ be the set of vertices presented so far (including $u$). 
Each set $m(x)$ for $x\in U-\{u\}$ is already known and the algorithm has to construct set $m(u)$. 
The construction is described below -- $m(u)$ is increased, one element at a time, until certain condition is satisfied.
During the process some other sets $m(x)$ may be decreased.

\begin{algorithm}
\label{alg}
\caption{$\alpha-\Alg{}(u)$}
\Let $m(u):=\emptyset$ \;
pick up at most $\alpha$ strongly available elements for $u$ and put it into $m(u)$ \; 
\While {there exists a vertex $e\in U$ that is ready for $u$ and satisfies 
   $$
    \abs{m(u)}  +2 \leq \abs{m(e)}\label{cond:balanced}
   $$ }
 { from the set of all such vertices pick $e$ with maximal size of $m(e)$ \;
   move one vertex available for $u$ from $m(e)$ to $m(u)$}
\end{algorithm}
The condition in line 2 guarantees that the size of $m(u)$ will be at most $\alpha$.
Note that $\alpha$-\Alg{} never leaves $m(u)$ empty if there exists an available element for $u$, so in this respect $\alpha$-\Alg{} can be considered as greedy. 
For $\alpha=1$ the above algorithm is just a simple greedy construction of a bipartite matching.

The following proposition describes the performance of $\alpha$-\Alg{}. 
\begin{prop}\label{prop:balanced_bound}
 The size $k$ of matching produced by $\alpha$-\Alg{} in a lazy matching game of size $n$ equals  $n-(x_0+x_1+\ldots+x_k)$
 for some pair $(k,(x_0,\ldots,x_k))$ satisfying~\eqref{ineq:main-res}.
\end{prop}
\begin{proof}
Consider an instance of the lazy matching game of size $n$ in which \Br produced graph $G=(U,D,E)$, and algorithm $\alpha$-\Alg{} constructed matching $m:U\to\mathcal{P}(D)$. 
Suppose that $N$ rounds have been played in the game.
\emph{Presenting time} of element $u\in U$ is the index of the round in which $u$ has been presented. 

We denote by $m^t:U\to\mathcal{P}(D)$, the (partial) matching constructed up to round $t$. 
In particular, for $u \in U$ that is presented in round $t_0$, we have $m^t(u)=\emptyset$ for $t< t_0$ and then $(m^{t_0}(u), m^{t_0+1}(u), \ldots,m^{N}(u))$ is a weakly decreasing sequence of sets with $m^{N}(u)=m(u)$. 

Let $X$ be the set of all vertices in $D$ such that $m(u)\cap X = \emptyset$ for each $u\in U$. 
The size of the matching produced be $\alpha$-\Alg{} is equal to the size of $\mathbf Y = \{u \in U : m(u)\neq \emptyset\}$.
For the proof of the proposition we need the following claims. 
\begin{claim}
\label{claim:d1}
  Suppose that  $x\neq y$, $m^{t_1}(x) \cap m^{t_2}(y)\neq\emptyset$ and $t_1 < t_2$, then $\abs{m^{t_1+1}(x)} \geq\abs{m^{t_2}(y)}$. 
\end{claim}
\begin{proof}
It is sufficient to verify the claim for $t_2=t_1+1$. 
It means that 
during round $t_2$ algorithm $\alpha$-\Alg{} removed one element from $m(x)$ and inserted it into $m(y)$. 
Let $d\in m^{t_1}(x)\cap m^{t_2}(y)$ be the last such element (see Figure~\ref{fig:common-element}).
This happens only when the condition from the line 3 of the algorithm  is satisfied and  $x$ is a vertex with maximum size of assigned set among vertices ready for $y$. 
Let $s$ be the size of the set assigned to $x$ at that moment  (in terms of listing~\ref{alg} it is $\abs{m(x)}$). 
Clearly $\abs{m^{t_1}(x)} \geq s$, since $\abs{m^{t_1}(x)}$ is the size of the set assigned to $x$ in the beginning of round $t_2$, and that set can only get smaller during the round. 
Also $s-1=\abs{m^{t_1+1}(x)}$ since element $d$ was the last one removed from $m(x)$.

The condition from the line 3 of $\alpha$-\Alg{} guarantees that the set that has just been increased has no more elements than the one that has been decreased. 
That property, and the fact that no vertex that was ready for $y$ had assigned set greater than $s$, gives $s > \abs{m^{t_2}(y)}$. The claim follows.
\end{proof}

\usetikzlibrary{arrows,calc,shapes,decorations.pathreplacing}
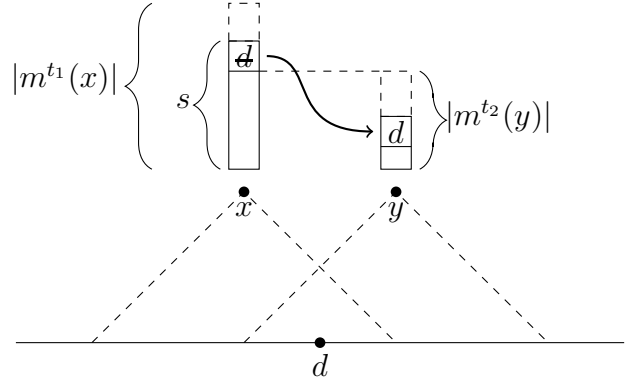
\begin{figure}
\begin{center}
\begin{tikzpicture}

\coordinate (d) at (5,0);
\coordinate (x) at (4,2);
\coordinate (y) at (6,2);

\fill (d) circle (2pt) node[below] {$d$};
\fill (x) circle (2pt) node[below] {$x$};
\fill (y) circle (2pt) node[below] {$y$};

\path[draw] (1,0) -- (9,0);
\path[draw,dashed] (2,0) -- (x) -- (6,0);
\path[draw,dashed] (4,0) -- (y) -- (8,0);

\draw (3.8,2.3) rectangle (4.2,4);
\draw[dashed] (3.8,4) rectangle (4.2,4.5);
\path (4,3.8) node {\st{$d$}};
\path[draw] (3.8,3.6) -- (4.2,3.6);

\draw (5.8,2.3) rectangle (6.2,3);
\draw[dashed] (5.8,3) rectangle (6.2,3.6);
\path[draw] (5.8,2.6) -- (6.2,2.6);
\path (6,2.8) node {$d$};

\draw[thick,->] (4.3,3.8) .. controls (5,3.8) and (4.5,2.8) .. (5.7,2.8);

\draw[decorate,decoration={brace,mirror,raise=6pt,amplitude=10pt}]    (3,4.5)--(3,2.3) ;
\draw[decorate,decoration={brace,mirror,raise=6pt,amplitude=10pt}]    (3.9,4)--(3.9,2.3);
\node at (3.2,3.2) {$s$};
\node[left] at (2.5,3.5) {$\abs{m^{t_1}(x)}$};

\draw[decorate,decoration={brace,raise=6pt,amplitude=10pt}]    (6.1,3.6)--(6.1,2.3);
\node[right] at (6.5,3) {$\abs{m^{t_2}(y)}$};

\path[draw,dashed] (4.2,3.6) -- (5.8,3.6);

\end{tikzpicture}
\caption{The move of element $d$ from $m(x)$ into $m(y)$.}
\label{fig:common-element}
\end{center}
\end{figure}

\begin{claim}\label{claim:els-above-X}
 If $t$ is the presenting time of $u$ and $N(u)\cap X \neq \emptyset$, then $\abs{m^t(u)}=\alpha$.
\end{claim}
\begin{proof}
 Let $d\in N(u)\cap X$. By the definition of $X$ element $d$ is strongly available for $u$ when $u$ is presented. 
But $d$ is not chosen in the line $2$ of $\alpha$-\Alg{}. 
It means there were at least $\alpha$ other strongly available elements for $u$ which were added to $m^t(u)$. 
Thus, indeed $\abs{m^t(u)}=\alpha$.
\end{proof}

\begin{claim}\label{claim:main-ineq}
For any subset $Y\subseteq \mathbf Y$  we have
$$(\abs{Q}-\abs{Y})(\mu - 1) + \abs{M}  \leq \abs{D-X},$$
where $\mu= \min\{\abs{m(y)} : y \in Y\}, M=\bigcup_{y\in Y} m(y)$ and $Q$ is the set of all vertices $q\in U$ for which $N(q) \cap (M\cup X) \neq \emptyset$. 
\end{claim}
\begin{proof}
The claim is obvious for $\mu=1$ since $M\cap X = \emptyset$.
We assume that $\mu>1$.
Let $i=\abs{Y}, s=\abs{Q}$ and let $(q_1, \ldots, q_s)$ be the enumeration of $Q$ for which the sequence of corresponding presenting times $(t_1, \ldots, t_s)$ is strictly decreasing.
It means that $q_1,\ldots,q_s$ are in the reverse of the arrival order.
For each $q_j$ we recursively define  set $$Z_j:=m^{t_j}(q_j)-(Z_1\cup\ldots\cup Z_{j-1})\subseteq D-X.$$ 
Observe that $\abs{Z_1}+\ldots+\abs{Z_s}\leq \abs{D-X}$ and to finish the proof it suffices to show that $\abs{Z_j}\geq \mu - 1$, for all $j$ ($1\leq j\leq s$), and $Z_j \supseteq m(q_j)$ for $q_j\in Y$. We check two possibilities, either $q_j\in Y$ or $q_j\notin Y$.
The first one is straightforward.
Set $m^t(q_j)$ can only get smaller, after  vertex $q_j$ has been presented, therefore we have $m(q_j)\subseteq m^t(q_j)$ for all $t\geq t_j$. 
In particular, for every $j'\leq j-1$, we have $t_{j'} \geq t_j$, hence $m(q_j) \cap Z_{j'} = \emptyset$. That gives $Z_j \supseteq m(q_j)$.

Suppose now, that $q_j\notin Y$ and let $m^{t_j}(q_j)=Z_j\cup R_j$, where $R_j\subseteq Z_1\cup\ldots\cup Z_{j-1}$. 
Assume also that $\abs{Z_j}<\alpha$ since otherwise $\abs{Z_j}=\alpha>\mu-1$.
We consider two cases:

\textbf{Case 1: $R_j=\emptyset$.}
By Claim~\ref{claim:els-above-X} and the definition of $q_j\in Q$, inequality $\abs{Z_j}<\alpha$ implies that $N(q_j)\cap M\neq\emptyset$.
It  means that some element $d\in M$ was available for $q_j$ at the time when $q_j$ was presented. 
Element $d$ must belong to $m^{t_j}(u)$ for some $u\in U$ presented earlier 
(otherwise the algorithm would put it in $m^{t_j}(q_j)$, but since after the game we have  $d\in M$, it would  imply that $d\in R_j$). 
Thus, there are two elements (possibly the same) $u\in Q$ and $y\in Y$ such that $d\in m^{t_j}(u) \cap m(y)$ and obviously $y$ is presented after $u$.
While $u$ is presented before $q_j$ element $y$ may be presented after $q_j$.
Regardless of whether or not $u$ and $y$ are the same
by Claim~\ref{claim:d1} we deduce $\abs{m^{t_j}(u)} \geq \abs{m(y)} \geq \mu $. 
On the other hand the algorithm in round $t_j$  did not choose element $d$ to be assigned to $q_j$, which means that at the end of the round the inequality in the line 3 of the algorithm was not satisfied. 
That means that $\abs{Z_j}=\abs{m^{t_j}(q_j)} \geq \abs{m^{t_j}(u)}-1 \geq \mu -1$.

\textbf{Case 2: $R_j \neq \emptyset$.} 
Let $t>t_j$ be the smallest number (the first moment) for which $R_j \cap m^t(q_j) = \emptyset$ (it is a straightforward consequence of the definition that such $t$ exists). 
Clearly  $\abs{Z_j} \geq  \abs{m^{t}(q_j)}$. 
Consider any $d\in m^{t-1}(q_j) \cap  R_j$ and note there is $l<j$ such that $d\in Z_l\subseteq m^{t_l}(q_l)$ with $t-1 < t_l$. 
By Claim~\ref{claim:d1} it means that $\abs{m^{t}(q_j)}\geq \abs{m^{t_l}(q_l)}$, thus $\abs{Z_j} \geq \abs{Z_l}$. 
Straightforward induction (with Case 1 as basis) gives $\abs{Z_j}\geq \mu-1$.
\end{proof}

We are ready to prove the proposition. 
Fix any optimal (maximum) matching in graph $G$ and let $F\subseteq D$ be the set of all elements in $D$ outside the  matching. 
Consider an enumeration $(y_1,\ldots, y_k)$ of  $\mathbf Y$ such that, for  $x_i=\abs{m(y_i)}-1$, we have $ x_1\geq x_2\geq \ldots\geq x_k\geq 0$. 
Let $x'=\abs{X-F}$, $f'=\abs{F-X}$ and $x_0 = x'-f'$.
Observe that $\abs{F} = f'+\abs{X}-x'$. 
It implies
\begin{align}\label{eq:size:D-X}
 \abs{D-X}=\abs{D}-\abs{X}= n + \abs{F}-\abs{X} = n-x_0.
\end{align}

To show that $(k, (x_0,x_1, \ldots,x_k))$ satisfy~\eqref{ineq:main-res}, fix $1\leq i\leq k$ and apply Claim~\ref{claim:main-ineq} for $Y=\{y_1,\ldots,y_i\}$.  
Then $\abs{M}=x_1+\ldots+x_i+i$ and  $\mu = x_i+1$. 
Recall that in the chosen optimal matching each vertex in $D-F$ has a unique match in $U$.
Therefore, $\abs{Q} \geq \abs{M-F}+\abs{X-F} = \abs{M-(F-X)}+\abs{X-F}  \geq \abs{M}-f'+x'$ 
since $M\cap X = \emptyset$,
and then
$$(\abs{M}-i+x_0)(\mu-1)+\abs{M}\leq \abs{D-X} \buildrel \eqref{eq:size:D-X}\over = n-x_0,$$ 
which  can be  rewritten into
\[ (x_0+x_1+\ldots+x_i)\cdot(1+x_i)\leq n-i.\]
Let $Q'$ be the set of all vertices $q\in U$ for which $N(q)\cap X \neq\emptyset$. 
Next, we define $s(q)$ as the set of all strong available elements assigned to $m(q)$ in the line $2$ of the algorithm. 
Observe that $s(q_1)$ and $s(q_2)$ are disjoint for distinct $q_1,q_2\in U$. 
By Claim~\ref{claim:els-above-X} we get $\abs{s(q)}=\alpha$ for each $q\in Q'$. 
Thus, $\alpha\abs{Q'}\leq n-x_0$ since $\bigcup_{q\in Q'}s(q) \subseteq D-X$ and by~\eqref{eq:size:D-X}.
Also, since each element in $D-F$ has a unique match in $U$, we have $\abs{Q'}\geq \abs{X-F}=x'\geq x_0$.  Therefore
$ (1+\alpha)x_0 \leq n. $

To finish the proof recall that $\bigcup_{y\in \mathbf Y} m(y) = D-X$.
Thus, $x_1+\ldots+x_k+k=n-x_0$ and consequently the size of the matching constructed by the algorithm equals $k = n-(x_0+x_1+\ldots+x_k)$.
Also, since $k$ cannot be larger then $n$ we have $x_0+\ldots+ x_k \geq0$.
\end{proof}

Combining Proposition ~\ref{prop:spoiler_bound} and Proposition~\ref{prop:balanced_bound} we finally get Theorem~\ref{thm:optimal-bal}.

\section{Competitiveness of $\alpha$-\Alg{} algorithm}
\label{sect:comp}

Let $\bal(\alpha,n)$ be the worst (minimum) size of matching constructed by $\alpha$-\Alg{} in any $\alpha$-lazy matching game of size $n$.
Competitive ratio of $\alpha$-\Alg{} is defined as $\bal(\alpha) = \liminf_{n\to\infty} \bal(\alpha,n)/n$.
Propositions~\ref{prop:spoiler_bound} and~\ref{prop:balanced_bound} imply that in order to determine $\bal(\alpha,n)$ it is enough to find a pair $(k,(x_0,\ldots, x_k))$ satisfying~\eqref{ineq:main-res} which maximizes $\sum_{i=0}^k x_i$.
Moreover, by Proposition~\ref{prop:positiv:x0} we can assume that in the maximizing solution we have $x_0= \lfloor \frac{n}{1+\alpha} \rfloor$.
From now on we consider $x_0$ in system~\eqref{ineq:main-res} as fixed together with $n$ and $\alpha$.
Suppose that pair $(k,(x_1,\ldots, x_k))$ satisfies~\eqref{ineq:main-res}.
Note that, for $i\geq 1$, if $x_i= x_{i+1}$ then the $(i+1)$-th inequality of system~\eqref{ineq:main-res} implies the $i$-th inequality.
That suggests another representation of the solutions.
For a pair $(k,x)$ satisfying~\eqref{ineq:main-res} with $x=(x_1, \ldots, x_k)$, let $Y(x)=(y_1, \ldots, y_m)$ be such that $m=1+x_1$ and $y_j= |\{i>0: 1+x_i= j\}|$.
Then, for every $i$ for which $x_{i+1}\neq x_i$, inequality 
\[
   (x_0+x_1+\ldots+x_i)(1+x_i) \leq n-i
\]
can be rewritten into 
\[
	(x_0 + (m-1) \cdot y_m + \ldots +(t-1)\cdot y_t) \cdot t \leq n- (y_m + \ldots + y_t),
\]
where $t=x_i+1$.
By the above discussion sequence $(y_1, \ldots, y_m)$ belongs to the image of $Y$ whenever it satisfies the following system of inequalities 
\begin{equation}\label{cond:prob:B}
   \eqB_{n,m}(x_0):\qquad t\cdot x_0 + \sum_{i=t}^m (1+(i-1)t)\vr{y_i} \leq n \quad \text{ for $t=1, \ldots,m$}
\end{equation}
Moreover, since $y_i\geq0$ then $m$-th inequality in~\eqref{cond:prob:B} implies
\begin{equation}\label{cond:m}
 n-mx_0\geq0.
\end{equation}
On the other hand, having any $m>0$ satisfying~\eqref{cond:m} and any solution $(y_1,\ldots,y_m)$ of~\eqref{cond:prob:B} one can easily
find (using the definition of $y_j$'s) a solution $(k,x)$ of~\eqref{ineq:main-res} such that $x_0+x_1+\ldots + x_k = x_0 + \sum_{i=1}^m(i-1)y_i$.

The above considerations are summarized in the following
\begin{prop}
 The minimal size of the matching constructed by $\alpha$-\Alg{} in all $\alpha$-lazy matching games of size $n$ is equal to
 $$\bal(\alpha,n)= n -x_0- \sup\left\{\sum_{i=1}^m (i-1) y_i \right\},$$
 where $x_0=\lfloor n/(1+\alpha)\rfloor$ and supreme is taken over all integers $m>0$ such that $n-mx_0\geq0$ and over all vectors $(y_1,\ldots,y_m)$ of nonnegative integers that satisfy $\eqB_{n,m}(x_0)$.
\end{prop}

\subsection{LP formulation.}
In order to maximize $\sum_{i=1}^m (i-1) y_i$ we consider the following linear program. We present the primal and the dual formulation.
\begin{itemize}
 \item Primal:
 
 \begin{tabular}{rll}
  maximize    & $\sum_{j=1}^m (j-1)y_j$\\
  subject to  &  $\sum_{j=i}^m (1+(j-1)i)y_j \leq n- ix_0$, & for $i=1,\ldots,m$\\
  and         &  $y_j \geq 0$, & for $j=1,\ldots,m$
 \end{tabular}

 \item Dual:
 
 \begin{tabular}{rll}
  minimize    & $\sum_{i=1}^m (n-ix_0)z_i$\\
  subject to  &  $\sum_{i=1}^j (1+(j-1)i)z_i \geq j-1$, & for $j=1,\ldots,m$\\
  and         &  $z_i \geq 0$, & for $i=1,\ldots,m$
 \end{tabular}
\end{itemize}
Let $P_i(y) = \sum_{j=i}^m (1+(j-1)i)y_j$ and $D_j(z) = \sum_{i=1}^j (1+(j-1)i)z_i$. 
Observe that both systems $(P_i(y) = n-ix_0)_{i=1, \ldots, m}$ and $(D_j(z) = j-1)_{j=1, \ldots, m}$ are quadratic and have unique solutions. 
By the Complementary Slackness Theorem these solutions are optimal if both are nonnegative vectors.

To solve the first system observe that $P_{i+2}(y)+P_{i}(y) - 2P_{i+1}(y) = (1-i+i^2)y_{i}-(1+i)^2y_{i+1}$. 
This, together with initial values for $y_m$ and $y_{m-1}$ gives 
\begin{align}
 & y_i= y_{i+1} \frac{(i+1)^2}{1-i+i^2} \quad \text{ for $i=1, \ldots,m-2$},\label{solution:yi} \\
 & y_m  =  \frac{n-m \cdot x_0}{1+m(m-1)},\quad \quad y_{m-1}= \frac{x_0+(m-1)y_m}{1+(m-1)(m-2)}.\label{solution:ym}
\end{align}
By~\eqref{cond:m} $y_m\geq0$ and thus all $y_i$ are nonnegative.

For the second system observe that $D_{j+1}(z)+D_{j-1}(z)-2D_j(z) = (1+j+j^2)z_{j+1}-(1-j)^2z_j$. Again, with the initial values for $z_1$ and $z_2$ we get
\begin{align*}
 & z_{j+1} = z_j \frac{(j-1)^2}{1+j+j^2}\quad\text{ for $j=2,\ldots,m-1$},\\
 & z_1 = 0, \quad\quad z_2 = 1/3.
\end{align*}
Therefore $z_i\geq0$ and indeed the above vectors $y$ and $z$ are optimal solutions for the primal and the dual LP.

To calculate the target function $\sum_{j=1}^m (j-1)y_j$ consider new variable $y_0$ defined by the additional condition: $P_0(y) = n$.
Equation~\eqref{solution:yi} still works for $y_0$. 
Also, after combining $P_0(y)$ with $P_1(y)$ we get $ \sum_{j=1}^m (j-1)y_j = y_0-x_0$.
Finally, using~\eqref{solution:yi} and~\eqref{solution:ym} and after some rearrangement we find 
\begin{equation}\label{solution:y0}
 x_0 + \sum_{j=1}^m (j-1)y_j = y_0 = \frac{(m-1) n+x_0}{ m} \prod_{i=1}^{m-1} \frac{i+i^2}{1+i+i^2}.
\end{equation}

\subsection{Lower bound.}
The solution for LP relaxation of~\eqref{cond:prob:B} may not be integer therefore with the solution~\eqref{solution:y0} we have only the following bound
\[
 \bal(\alpha,n) \geq n-\sup\{y_0 : n-mx_0\geq0\}.
\]

Let $F(z,m)=\frac{(m-1)+z}{m} \prod_{i=1}^{m-1} \frac{i+i^2}{1+i+i^2}$. 
Then, for fixed $m$, we get $y_0= n \cdot F(x_0/n,m)$. 
Observe that function $F(z,m)$ is increasing with $m$, for $m\leq 1/z$ and $F(0,m)$ increases indefinitely with $m$.
Therefore for $x_0 = \lfloor n / (1+\alpha)\rfloor > 0$, we have
\begin{align}
	\bal(\alpha,n)/n &\geq 1-\max_{m\leq n/x_0}F(x_0/n,m)\nonumber\\
		        &= 1-F(x_0/n,\lfloor n/x_0 \rfloor ) \buildrel n\to\infty\over\longrightarrow 1 - F(1/(1+\alpha),1+\alpha),\label{eq:balLB}
\end{align}
and for $x_0 = 0$ (it happens when $\alpha\geq n)$ the bound is
\begin{align}\label{eq:inftyLB}
  \bal(\alpha,n)/n &\geq 1-\lim_{m\to\infty}F(0,m) =1-\prod_{i=1}^{\infty}\frac{i+i^2}{1+i+i^2} =
  1-\frac{\pi}{\cosh\tfrac{\sqrt{3}}{2}\pi}.
\end{align}

\subsection{Upper bound.}
Consider $m=\min(\lfloor n/x_0 \rfloor,\ln n )$ assuming $n/x_0$ is greater then $\ln n$ when $x_0=0$,
and let $(y_1, \ldots, y_m)$ be the optimal rational solution of $\Psi_{n,m}(x_0)$.
Let $v=(v_1, \ldots, v_m)$ be such that $v_i= \lfloor y_i \rfloor$.
Vector $v$ contains only nonnegative, integer entries and $v\leq y$. 
The shape of system $\Psi_{n,m}(x_0)$ guarantees that $v$ also satisfies $\Psi_{n,m}(x_0)$.
Finally we have
\[
	x_0 + \sum_{i=1}^m (i-1) v_i >  x_0 + \sum_{i=1}^m (i-1) y_i - \sum_{i=1}^{m} (i-1).
\]
By the definition of $m$ we know that $m = o(n)$.
Therefore by Proposition~\ref{prop:spoiler_bound} we get
\begin{align*}
	\bal(\alpha,n) &< n \cdot (1-F(x_0/n,m )) + m(m-1)/2
	                =n \cdot (1-F(x_0/n,m )) + o(n)
\end{align*}
Hence, for $x_0=\lfloor n / (1+\alpha) \rfloor > 0$ the bound is
\begin{align}\label{eq:balUB}
 \bal(\alpha,n)/n  < 1-F(x_0/n,\lfloor n/x_0 \rfloor ) + o(1) \buildrel n\to\infty\over\longrightarrow 1 - F(1/(1+\alpha),1+\alpha),
\end{align}
and for $x_0 = 0$ we have
\begin{align}\label{eq:inftyUB}
 \bal(\alpha,n)/n < 1 - F(0,\ln n) + o(1) \buildrel n\to\infty\over\longrightarrow 1-\prod_{i=1}^{\infty}\frac{i+i^2}{1+i+i^2} =
  1-\frac{\pi}{\cosh\tfrac{\sqrt{3}}{2}\pi}.
\end{align}

\subsection{Competitive ratio.}
When $\alpha$ is finite, by~\eqref{eq:balLB} and~\eqref{eq:balUB} the ratio of $\alpha$-\Alg{} is equal to  $\bal(\alpha)= 1 - F(\frac{1}{\alpha+1},\alpha + 1)$.
For the case $\alpha\to\infty$ note that $F(1/(1+\alpha),1+\alpha)\to {\pi}/{\cosh\tfrac{\sqrt{3}}{2}\pi}$.
Thus, by all~(\ref{eq:balLB}-\ref{eq:inftyUB}) we get $\bal(\alpha) = 1-{\pi}/{\cosh\tfrac{\sqrt{3}}{2}\pi}$.
This finally proves Theorem~\ref{thm:ratio}.

\section{Conclusion and remarks}\label{sec:remarks}

In the classical on-line matching problem, randomized approach has a big advantage over the deterministic one. 
This paper shows that the lazy method moves forward deterministic bounds. 
It is interesting to know what can be achieved with randomization of the lazy technique. 
\begin{problem}
 What is the competitive ratio of the randomized version of the lazy matching problem?
\end{problem}

The authors of~\cite{KalKir2000} consider a variant (called $b$-matching) where each server can realize up to $b$ tasks. 
Competitive ratio of their optimal algorithm approaches $1-\frac{1}{e}$ with $b\rightarrow \infty$, which is a barrier for any randomized matching algorithm (see~\cite{KVV90}). 
We expect that the lazy method  is capable of breaking $1-\frac{1}{e}$ limit in case of $b$-matching.
\begin{problem}
 What is the competitive ratio of the lazy $b$-matching problem?
\end{problem}

\bibliographystyle{plain} \bibliography{adaptm}

\end{document}